\documentclass[reqno]{amsproc}

\usepackage{amstext,amsmath,amssymb,amsfonts}
\usepackage[latin1]{inputenc}
\usepackage{epsfig}
\usepackage{hyperref}
\usepackage{color}

\usepackage{latexsym}

\newtheorem{theorem}{Theorem}

\newcommand{\bea}{\begin{eqnarray}}
\newcommand{\eea}{\end{eqnarray}}
\newcommand{\beq}{\begin{equation}}
\newcommand{\eeq}{\end{equation}}
\newcommand{\enn}{\nonumber \end{equation}}

\title[Solving Navier-Stokes equations coupled with a heat transfer equation]
{Solving Navier-Stokes equations coupled with a heat transfer equation using Bagarello's approach and the Hankel transform\footnote{Preprint: ICMPA-MPA/2015/11 }}

\author{Mahouton Norbert Hounkonnou}
\address[M.N.H.]{International Chair in Mathematical Physics and Applications,
ICMPA-UNESCO Chair, 072 BP 50, Cotonou, Rep. of Benin}
\email{norbert.hounkonnou@cipma.uac.bj}

\author{Villevo Adanhounm\`e}
\address[V.A.]{International Chair in Mathematical Physics and Applications,
ICMPA-UNESCO Chair, 072 BP 50, Cotonou, Rep. of Benin}
\email{adanhounmvillvo@yahoo.fr}%

\author{Jean Ghislain Compaor\'e}
\address[J.G.C.]{International Chair in Mathematical Physics and Applications,
ICMPA-UNESCO Chair, 072 BP 50, Cotonou, Rep. of Benin}
\email{ghislaincompaore@yahoo.fr}

\begin{document}

\maketitle 
 \begin{abstract}
  In this paper, the dynamics of an incompressible fluid  in a bounded connected domain,  described by Navier-Stokes equations coupled with a heat transfer equation, is investigated by a method inspired from the non-commutative strategy developed by Bagarello,  
(see {\it Int. Jour. of Theoretical Physics}, \textbf{43}, issue 12 (2004), 
p. 2371 - 2394).
The solution of involved systems of partial differential equations is derived with the help of the unbounded self-adjoint densely defined Hamiltonian operator of the physical model and the Hankel transform.
\\
Keywords : Heat transfer equation, Navier-Stokes equations, Bagarello's operator method, Hankel transform.
\end{abstract}
\today

\tableofcontents

\section{ Introduction}
 Various models of motion  are used to describe  physical
 phenomena in fluid mechanics. See for instance  the work based on Navier-Stokes equations reported in Refs.\cite{Eglit, Lu}. Unfortunately, due to the complexity of the involved
 systems of differential equations, most of these models are not exactly solvable. This explains the massive  use of 
approximation and numerical schemes to obtain information on the flow phenomena such as the velocity distribution, flow pattern, pressure losses, and so on. To cite a few, see, e. g.,  \cite{Kharrat} -  \cite{Korn}.

 In this paper, we show that the problem of the dynamics of an incompressible fluid in a bounded connected domain, described by   Navier-Stokes equations coupled with a heat transfer equation, can be analytically solved combining both the operator method developed by Bagarello \cite{Bagarello97} and the Hankel transform. 

 The  paper is organized as follows. In   section 2, 
 the considered model is described. In section 3, we obtain the exact solution to the
heat transfer equation by the Bagarello's non-commutative approach. The 
 Navier-Stokes equations are solved by using the Hankel transform. Then follow the concluding remarks  in section 4.

\section{Mathematical formulation}
We consider cylindrical coordinates $x=(r,\varphi, z)$ relative to the orthogonal basis $(e_r, e_{\varphi}, e_z)$ and
  incompressible viscous fluid flow with the $z$-axis  symmetry.  
The fluid layer is confined between two parallel plates of horizontal lengths $L_x$ and $L_y$ separated by a vertical $(z)$ distance $(d)$. The no-slip upper and lower plates are held at fixed temperatures $T_1$ and $T_2$ respectively. A uniform volumetric heat flux $H$ (with units power/volume ) is pumped into the layer. The governing equations of such a flow of an incompressible fluid in the standard Boussinesq approximation are Eglit\cite{Eglit} ,Lu\cite{Lu}
\begin{eqnarray}
\dfrac{\partial\bold u}{\partial t} +\bold u.\nabla \bold u&=& -\nabla p+ \nu\nabla^2\bold u + (g\alpha T)\widehat{\bold k}\quad \rm{in}\quad \Omega^{\tau}\equiv \Omega\times (0, \tau),\label{1}\\
\dfrac{\partial T}{\partial t} +\bold u. \nabla T&=&\kappa\nabla^2 T +\gamma \quad \rm{in}\quad \Omega^{\tau}, \label{2}\\
\nabla.\bold u&=&0 \quad \rm{in}\quad \Omega^{\tau},\label{3}
\end{eqnarray}
with the initial and boundary conditions
\begin{eqnarray}
\bold u|_{t=0}&=& \bold u_0,\,\,\, T|_{t=0}= T_0\quad \rm{in}\quad \Omega, \label{4}\\
\bold u&=&0 \quad \rm{in}\quad \Gamma^{\tau}\equiv \Gamma\times (0,\tau),\label{5}\\
T|_{z=0}&=&T_1,\,\,\,T|_{z=d}= T_2, \label{6}
\end{eqnarray}
where 
$\Omega $ is a closed bounded set of $\mathbf{R}^3$ with the Lipschitz continuous boundary $\Gamma$, $\bold u(t, x)= (u(t, r, z),0, v(t, r, z))$ is the velocity of the fluid, $p$ is the fluid pressure, $T$ is the temperature,  $\widehat{\bold k}=(0, 0, 1)$, 
$\nu$ is the viscosity, $g$ is the acceleration of gravity along the $z$-axis (in the $-\widehat{\bold k}$ direction), $\alpha$ is the thermal expansion coefficient, $\kappa$ is the thermal diffusion coefficient.  $\gamma=H/(\rho c);$  $\rho$ is the density and $c$ is the specific heat capacity of the fluid. 
\section{Analytical solution}
In this section, we investigate the analytical solutions to the above described model based on  Navier-Stokes equations using an operator approach and the Hankel transform. 

As matter of commodity, we adopt  $d^2/\kappa$ as the unit of time, $d$ as the unit of length and $\gamma d^2/\kappa$ as the unit of temperature. The  equations can then be put into the following dimensionless form: 
\begin{eqnarray}
Pr^{-1}\bigg(\dfrac{\partial\bold u}{\partial t} +\bold u.\nabla \bold u\bigg)&=& -\nabla p+ \nabla^2\bold u + R T \widehat{\bold k},\label{7}\\
\dfrac{\partial T}{\partial t} +\bold u. \nabla T&=& \nabla^2 T + 1, \label{8}
\end{eqnarray}
where $Pr=\nu/\kappa$ is the Prandtl number and $R=g\alpha d^5\gamma/(\kappa^2\nu)$ is the heat Rayleigh number. The boundary conditions now read:
\begin{eqnarray}
\bold u=0 \quad \rm{in}\quad \overline{\Gamma}^{\tau}\equiv \overline\Gamma\times (0,\tau),\,\,
T|_{z=0}=\widetilde T,\,\,\,T|_{z=1}= 0, \label{9}
\end{eqnarray}
where $ \overline\Gamma$ is the boundary obtained from $\Gamma$ by replacing $z=d$ by $z=1$,
$\widetilde T=(\kappa/\gamma d^2)(T_1-T_2)$. 

We suppose that the solution of Cauchy problem exists and is 
unique. Then, we solve the problem by Bagarello's approach \cite{Bagarello97}, 
based on a suggestion drawn from quantum mechanics. It is based on the following: given a 
quantum mechanical system $S$ and the related set of observables 
$O_S,$ i. e., the set of all the self-adjoint bounded (or more often unbounded) 
operators describing $S$, the evolution of any observable $Y\in O_S$ 
satisfies the Heisenberg equation of motion (HOEM):
\begin{eqnarray}
\dfrac{\partial}{\partial t} Y (t, x) = i[H,Y(t, x)]. \end{eqnarray}
Here $[A,B] = AB- BA$ is the commutator between $A, B\in O_S$; $H$ is 
assumed to be a densely defined self-adjoint Hamiltonian operator of 
the system acting on some Hilbert space $\mathcal{H}$ while the initial 
condition $Y^0=h(x)$ is considered as an operator acting on the same
Hilbert space $\mathcal{H}$.
Then, the following statement holds:
  \begin{theorem}. A formal solution of the HOEM is 
$ Y(t,x)=e^{itH}Y^0e^{-itH},$ where $Y^0$ is the initial condition of 
$Y(t,x)$ and $H$ does not depend explicitly on time.
 Furthermore, if $H$ is bounded, we get $Y (t,x) =\sum_{k=0}^{\infty}
\dfrac{ [it)^k}{k!}[H,Y^0]_k,$ where $[A,B]_k$ is the multiple commutator 
defined recursively as : $[A,B]_0 = B; [A,B]_k = [A,[A,B]_{k-1}],$ etc.
\end{theorem} 
\begin{proof} We have  \begin{eqnarray*}
\dfrac{\partial}{\partial t} Y (t,x)&=&\dfrac{\partial}{\partial t}[e^{itH}]Y^0e^{-itH}+
e^{itH}Y^0\dfrac{\partial}{\partial t}[e^{-itH}]\Longleftrightarrow\\
\dfrac{\partial}{\partial t} Y (t,x)&=& iHe^{itH}Y^0e^{-itH}+
e^{itH}Y^0(-iH)e^{-itH}\Longleftrightarrow\\
\dfrac{\partial }{\partial t} Y (t,x)&=& i[H, Y (t,x)] \end{eqnarray*}
If $H$ is bounded, $Y(t,x)=e^{itH}Y^0e^{-itH}$ can be expanded as
\begin{eqnarray*} Y (t,x) =\sum_{k=0}^{\infty} \dfrac{(it)^k}{k!}[H,Y^0]_k.
\end{eqnarray*} \end{proof}$\Box$

Consider the finite Hankel transform of $\Phi$ defined by \cite{Korn}
\begin{eqnarray}
\overline\Phi(t, \mu_n )&=&\int_0^{\xi_0}\xi\Phi( t,\xi )J_{3/2}(\mu_n\xi)d\xi,\\
\overline\Phi(t, \mu_n)|_{t=0}&=&\int_0^{\xi_0}\xi^{5/2}\Omega_0(\xi)J_{3/2}(\mu_n\xi)d\xi,\\
\Phi(t, \xi )&=&\dfrac{2}{\xi^2_0}\sum_{n\geq 1}\dfrac{J_{3/2}(\mu_n\xi)}{[J^{\prime}_{3/2}(\mu_n\xi_0)]^2}\overline{\Phi}(t, \mu_n ),
\end{eqnarray}
where $J_{3/2}$ is the Bessel function of the first kind with the $3/2$-order, $\mu_1, \mu_2, \mu_3,\ldots,\mu_n,\ldots$ are the positive roots of the equation
\begin{equation}
J_{3/2}(\mu\xi_0)= 0, \,\,\,\xi_0=\rm{const}.
\end{equation}
Defining the stream function $\psi$ and the component of vorticity by $$u=-\dfrac{1}{r}\dfrac{\partial\psi}{\partial z},\,\,\, v= \dfrac{1}{r}\dfrac{\partial\psi}{\partial r},\,\,\,\omega=\dfrac{\partial u}{\partial z}-\dfrac{\partial v}{\partial r}
$$ and taking the operator $\rm{rot}$ of the equation $(\ref{7})$, we obtain \cite{Eglit}
\begin{eqnarray}
Pr^{-1}\bigg(\dfrac{\partial\omega }{\partial t}+\dfrac{\partial (\omega/r, \psi)}{\partial (r, z)}\bigg)&=&\dfrac{1}{r^2}\Bigg[\dfrac{\partial }{\partial r}\bigg(r^3\dfrac{\partial (\omega/r)}{\partial r}\bigg)+\dfrac{\partial }{\partial z}\bigg(r^3\dfrac{\partial (\omega/r)}{\partial z}\bigg)\Bigg]
- R\dfrac{\partial T}{\partial r}\label{10},\\
\dfrac{\partial T}{\partial t} +\dfrac{1}{r}
\dfrac{\partial ( \psi, T)}{\partial (r, z)}&=& 
\dfrac{1}{r}\dfrac{\partial}{\partial r}\bigg(r\dfrac{\partial T}{\partial r}\bigg) +\dfrac{\partial^2 T}{\partial z^2}+ 1,\label{11}
\end{eqnarray}
where 
\begin{eqnarray}
\dfrac{\partial ( A, B)}{\partial (r, z)}&:=&\dfrac{\partial A}{\partial r}\dfrac{\partial B }{\partial z}-\dfrac{\partial A }{\partial z }\dfrac{\partial B }{\partial r}.
\end{eqnarray}
Using the transformation $\xi=\sqrt{r^2+z^2}$ and setting 
$\Omega(\xi,t)= \omega(r,z, t)/r,\,\,\,\widehat T(\xi, t)=T(r,z,t)- t$, the equations (\ref{10}) and (\ref{11}) become
\begin{eqnarray}
{Pr}^{-1}\dfrac{\partial\Omega }{\partial t}&=&
\dfrac{\partial^2\Omega}{\partial\xi^2}+\dfrac{4}{\xi}\dfrac{\partial\Omega}{\partial \xi}- \dfrac{R}{\xi}\dfrac{\partial\widehat T}{\partial\xi},\label{12}\\
\dfrac{\partial(\xi\widehat T)}{\partial t} &=&\dfrac{\partial^2(\xi\widehat T)}{\partial\xi^2}.\label{13}
\end{eqnarray}
Using Bagarello's approach, the solution of the equation (\ref{13}) satisfying the initial condition  can be expressed as: 
\begin{eqnarray}
\xi\widehat T(\xi, t)&=&\sum_{n\geq 0}\dfrac{(it)^n}{n!}(-i)^n\Delta^n [\xi T_0(r, z)]
=\sum_{n\geq 0}\dfrac{ t^n}{n!}\Delta^n[\xi T_0(r,z)]\Longleftrightarrow\\
T(\xi, t)&=& t+\dfrac{1}{\xi}\sum_{n\geq 0}\dfrac{t^n}{n!}\Delta^n[\xi T_0(r,z)],
\end{eqnarray}
where $\Delta$ is the Laplacian operator.
Using the change of variable $\Omega=\xi^{-3/2}\Phi,\,\,\Phi=\Phi(t, \xi )$ in the equation (\ref{12}) we obtain
\begin{eqnarray}\label{eq22}
{Pr}^{-1}\dfrac{\partial\Phi}{\partial t}&=&
\dfrac{\partial^2\Phi}{\partial\xi^2}+\dfrac{1}{\xi}\dfrac{\partial\Phi}{\partial\xi}-\dfrac{(3/2)^2}{\xi^2}\Phi -R\xi^{1/2}\dfrac{\partial\widehat T}{\partial\xi}\label{14},\\
\Phi|_{t=0}&=&\xi^{3/2}\Omega_0(\xi),
\end{eqnarray}
where $\Omega_0(\xi)=(\omega/r)|_{t=0}$.\\
Applying the Hankel transformation to the equation (\ref{eq22}) and taking account of the boundary condition (\ref{9}), i. e. $\Phi(t, \xi_0 )=0$, we get
\begin{eqnarray}
Pr^{-1}\dfrac{d\overline\Phi}{dt}&=& -\mu^2_n\overline{\Phi}-\overline{\Upsilon}(t, \mu_n )\label{15},\\
\overline\Phi|_{t=0}&=&\overline\Phi(0, \mu_n ), \label{16}
\end{eqnarray}
and
\begin{eqnarray}
\overline\Upsilon(t, \mu_n )= R \int_0^{\xi_0}
\xi^{3/2}\dfrac{\partial\widehat T}{\partial\xi}(t,\xi )J_{3/2}(\mu_n\xi)d\xi.
\end{eqnarray} 
Thus the solution of equation (\ref{15}) satisfying the initial condition (\ref{16}) can be expressed as 
\begin{eqnarray}
\overline\Phi(t, \mu_n )= \overline\Phi(0, \mu_n )e^{-Pr\mu^2_n t}- Pr \int_0^t e^{-Pr\mu^2_n( t-\tau)}
\overline\Upsilon(\tau, \mu_n)d\tau
\end{eqnarray}
and 
\begin{eqnarray}
\Omega(t, \xi )=\dfrac{2}{\xi^2_0}\xi^{-3/2}\sum_{n\geq 1}\dfrac{\overline\Phi(t, \mu_n )}{[J^{\prime}_{3/2}(\mu_n\xi_0)]^2}J_{3/2}(\mu_n\xi).
\end{eqnarray}
Taking account of the expression of $\omega$
we obtain the equation with the unknown $\psi$
\begin{eqnarray}
 \dfrac{\partial^2\psi}{\partial\xi^2} 
 =\dfrac{-2r^2}{\xi^2_0}\xi^{-3/2}\sum_{n\geq 1}
  \dfrac{\overline\Phi(t, \mu_n )}{[J^{\prime}_{3/2}(\mu_n\xi_0)]^2}J_{3/2}(\mu_n\xi)
  \end{eqnarray}
and the derivative of stream function $\psi$, components $u$ and $v$ can be defined by
\begin{eqnarray*}
 \dfrac{\partial\psi}{\partial\xi}&=&\dfrac{-2r^2}{\xi^2_0}\int_0^\xi \eta^{-3/2}\sum_{n\geq 1}
 \dfrac{\overline{\Phi}(t, \mu_n )}{[J^{\prime}_{3/2}(\mu_n\xi_0)]^2}J_{3/2}(\mu_n\eta)d\eta, \\
 u(r,z,t)&=& \dfrac{2rz}{\xi\xi^2_0} \int_0^\xi\eta^{-3/2}\sum_{n\geq 1}
\dfrac{\overline{\Phi}(t, \mu_n )}{[J^{\prime}_{3/2}(\mu_n\xi_0)]^2}J_{3/2}(\mu_n\eta)d\eta, \\
 v(r,z,t)&=& \dfrac{-2r^2}{\xi\xi^2_0}\int_0^\xi\eta^{-3/2}\sum_{n\geq 1}
\dfrac{\overline{\Phi}(t, \mu_n )}{[J^{\prime}_{3/2}(\mu_n\xi_0)]^2}J_{3/2}(\mu_n\eta)d\eta.
\end{eqnarray*}
Using Hankel transformation and Bagarello'approach, we provide analytical solutions to the nonlinear 
Navier-Stokes equations and to heat transfer equation  (\ref{1}) and (\ref{2}) satisfying the initial and boundary conditions (\ref{3})-(\ref{6}).We can state the main result as:
\begin{theorem}
The velocity components $u,\,\,v$, the vorticity component $\omega$ and the temperature $T$ of the incompressible fluid flow described by the Navier-Stokes equations coupled with the heat transfer equation are defined by 
\begin{eqnarray}
 u(r,z,t)&=& \dfrac{2rz}{\xi\xi^2_0} \int_0^\xi\eta^{-3/2}\sum_{n\geq 1}
\dfrac{\overline{\Phi}(t, \mu_n )}{[J^{\prime}_{3/2}(\mu_n\xi_0)]^2}J_{3/2}(\mu_n\eta)d\eta, \\
 v(r,z,t)&=& \dfrac{-2r^2}{\xi\xi^2_0}\int_0^\xi\eta^{-3/2}\sum_{n\geq 1}
\dfrac{\overline{\Phi}(t, \mu_n )}{[J^{\prime}_{3/2}(\mu_n\xi_0)]^2}J_{3/2}(\mu_n\eta)d\eta, \\
 \omega(r,z, t)&=&\dfrac{2r}{\xi^2_0}\xi^{-3/2}\sum_{n\geq 1}
\dfrac{\overline{\Phi}(t, \mu_n )}{[J^{\prime}_{3/2}(\mu_n\xi_0)]^2}J_{3/2}(\mu_n\xi), \\
T(r, z, t)&=& t+\dfrac{1}{\xi}\sum_{n\geq 0}\dfrac{t^n}{n!}\Delta^n[\xi T_0(r,z)].
\end{eqnarray}
\end{theorem}
 \section{Concluding remarks}
 In this paper, we  succeeded in solving the problem of  fluid dynamics described by the Navier-Stokes equations coupled with the heat transfer equation. These equations were transformed  in forms suitable to apply the Bagarello's operator approach and the Hankel transform.   Especially, we computed the velocity, the vorticity and the temperature of the fluid flow  characterizing  the fluid dynamics.

\end{document}